\documentclass[runningheads]{llncs}

\usepackage{graphicx}
\usepackage{color}
\usepackage{amssymb}
\usepackage{amsmath}
\usepackage{xfrac}

\usepackage{tikz}
\usepackage{tikz-cd}
\usepackage[hidelinks]{hyperref}
\usepackage{mathtools}
\usepackage{xcolor}
\usepackage{algpseudocodex}
\usepackage{algorithm}
\usepackage{enumerate}
\usepackage{nicematrix}
\usepackage{adjustbox}
\usepackage[title,toc,titletoc,header]{appendix}
\usepackage{soul} % Only for strikethrough \st

\usepackage[colorinlistoftodos,prependcaption,textsize=tiny]{todonotes}

\newcommand {\mm}[1]   {\ifmmode{#1}\else{\mbox{\(#1\)}}\fi}

\newcommand{\Rspace}        {\mm{{\mathbb R}}}

\newcommand{\Zspace}        {\mm{{\mathbb Z}}}

\newcommand{\MSAxis}[1]     {\mm{{\mathcal Mid}}{({#1})}}

\newcommand{\Birth}[1]      {\mm{{\rm bth}{({#1})}}}
\newcommand{\Death}[1]      {\mm{{\rm dth}{({#1})}}}
\newcommand{\Premed}[3]     {\mm{{\rm Pre}_{#1}}{({#2},{#3})}}
\newcommand{\Postmed}[3]    {\mm{{\rm Post}_{#1}}{({#2},{#3})}}

\newcommand{\Betti}[1]      {\mm{{\beta}_{#1}}}

\newcommand{\Edist}[2]      {\mm{\|{#1}-{#2}\|}}

\newcommand{\Skip}[1]       {}
\newcommand{\HE}[1]         {{\textcolor{blue}{{{\rm {#1}}}}}}
\newcommand{\ES}[1]         {{\textcolor{cyan}{{{{\rm #1}}}}}}

\newcommand{\bigo}[1]     {{\mathcal O}{\left({#1}\right)}}

\begin{document}

%% \title{The Mid-sphere Variant of the Medial Axis Transform}
%% \title{The Mid-sphere Relative of the Medial Axis Transform}
\title{The Mid-sphere Cousin of the Medial Axis Transform}

\author{
Herbert Edelsbrunner\
\inst{1}\
\orcidID{0000-0002-9823-6833} \and \\
Elizabeth Stephenson\
\inst{1,2}\
\orcidID{0000-0002-6862-208X} \and  \\
Martin Hafskjold Thoresen\
\inst{3}\
\orcidID{0009-0007-8322-6492}
}
\authorrunning{H.\ Edelsbrunner, E.\ Stephenson and M.H.\ Thoresen}

\institute{ISTA, 3400 Klosterneuburg, Austria, 
\email{edels@ist.ac.at}
\and Orteliu, Oslo, Norway
\email{elizasteprene@gmail.com} 
\and Independent, Drammen, Norway
\email{m@mht.wtf}
}

% \email{$^3$edels@ist.ac.at}
% \date{}
% \ccsdesc[100]{Theory of computation~Computational geometry}

\maketitle

\begin{abstract}
  The \emph{medial axis} of a smoothly embedded surface in $\Rspace^3$ consists of all points for which the Euclidean distance function on the surface has at least two global minima.
  We generalize this notion to the \emph{mid-sphere axis}, which consists of all points for which the Euclidean distance function has two interchanging saddles that swap their partners in the pairing by persistent homology.
  It offers a discrete-algebraic multi-scale approach to computing ridge-like structures on the surface.
  As a proof of concept, an algorithm that computes stair-case approximations of the mid-sphere axis is provided.


\end{abstract}

%%%%%%%%%%%%%%%%%%%%%%%%%%%%%%%%%%%%%%%%%%%%%%%%%%
%%%%%%%%%%%%%%%%%%%%%%%%%%%%%%%%%%%%%%%%%%%%%%%%%%
\section{Introduction}
\label{sec:1}
%%%%%%%%%%%%%%%%%%%%%%%%%%%%%%%%%%%%%%%%%%%%%%%%%%
%%%%%%%%%%%%%%%%%%%%%%%%%%%%%%%%%%%%%%%%%%%%%%%%%%

In the late 1960s, Harry Blum~\cite{Blu67} revolutionized the automated classification of shapes that typically arise in biology and related fields by suggesting the medial axis as a skeleton fit to describe amorphous blobs.
The concept was introduced in the smooth setting by Federer~\cite[Definition~4.1 on page~432]{Fed59} almost a decade earlier, but his contribution was largely overlooked.
Since then, the medial axis has been used as a tool in shape classification, image analysis, animation, computer graphics, and other fields that benefit from the feature-preserving simplification of unwieldy objects; see e.g.\ \cite{BTG95,TaHe03,YaYa18}. 

\smallskip
Given a closed surface, $M \subseteq \Rspace^3$, the \emph{medial axis} is the set of centers of spheres that touch $M$ in at least two points and do otherwise not intersect $M$.
Slightly more elaborately, the \emph{medial axis transform} is the same set of points paired with the radii at which the spheres touch $M$.
Requiring instead that $M$ lies inside the spheres, we get an apparently less useful complementary concept, which for the purposes of this paper we refer to as the \emph{circum-sphere axis transform} of $M$.
We observe that there is an appealing third option: the \emph{mid-sphere axis transform}, which consists of centers and radii of spheres that touch $M$ in at least two points, and for any two there is a closed curve on $M$ inside the sphere, and a closed curve on the sphere together with a filling surface both inside $M$; see Figure~\ref{fig:EllipsoidA} for an illustration.
\begin{figure}[hbt]
  \vspace{-0.1in}
    \includegraphics[width=.60\textwidth]{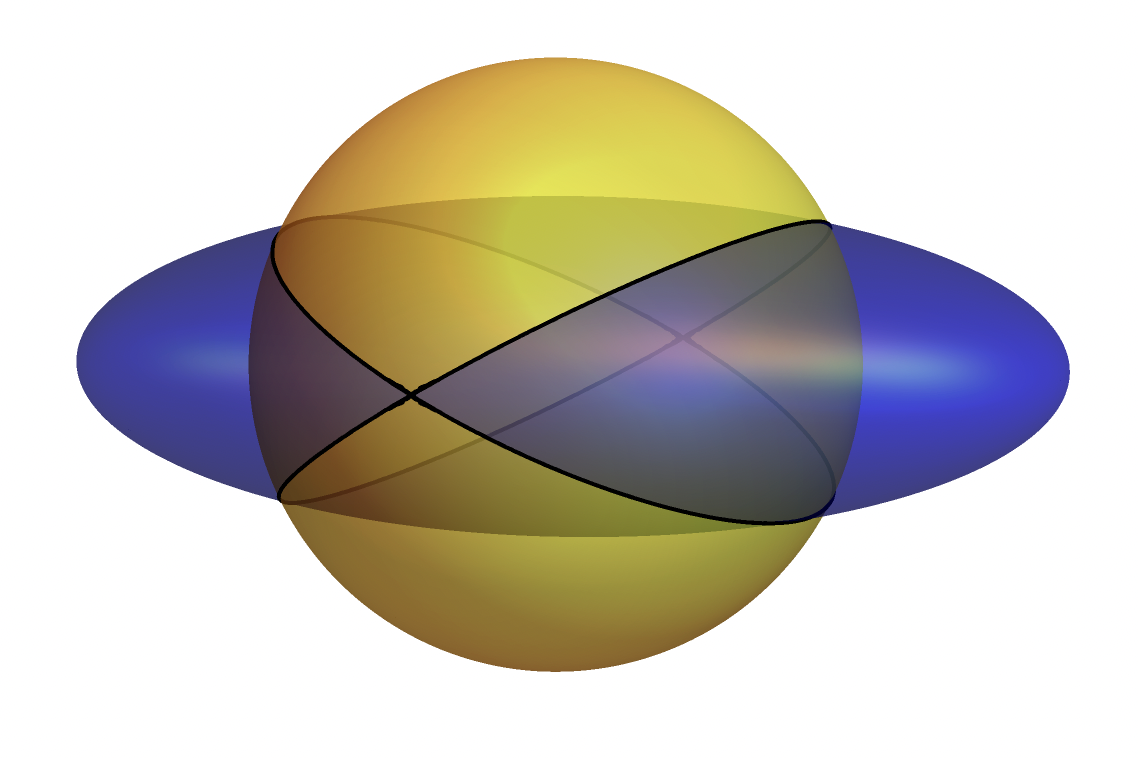}
    \vspace{-0.3in}
    \centering
    \caption{\footnotesize A mid-sphere touching an ellipsoid in two antipodal points.
    The vertical plane that passes through these two points intersects the ellipsoid in a curve that lies inside the sphere,
    and the horizontal plane that passes through the same two points intersects the sphere in a great-circle filled by a disk that both lie inside the ellipsoid.}
    \label{fig:EllipsoidA}
\end{figure}
\footnote{We borrow the terminology from $3$-dimensional geometry, where the \emph{mid-sphere} of a convex polytope is a sphere that touches every one of its edges.
The sphere thus intersects every facet in an inscribed circle.
By the Circle Packing Theorem of Koebe~\cite{Koe36}, Andreev~\cite{And70}, and Thurston~\cite[Chapter 13.6]{Thu02}, every $2$-connected planar graph has a geometric realization as the edge skeleton of a convex polytope that has a mid-sphere.
Curiously, this is not true if we require an in-sphere or a circum-sphere.}

\smallskip
All three axes are subsets of the \emph{symmetry set} of $M$ defined as the set of centers of spheres that touch $M$ in at least two points.
Compared to this ``hopelessly complicated object'' \cite{BGG85}, the three axes are sparse skeleta that specialize in bringing out particular types of features.
For example, the mid-sphere axis is sensitive to ridge-like features, in which a point $x \in M$ belongs to a curve called a \emph{ridge} if it locally minimizes or maximizes the curvature along a principle curvature direction; see e.g.\ Chapter 11 on ``Ridges and Ribs'' in \cite{Por01}.
This definition is in terms of derivatives and thus hyper-sensitive to local noise.
We propose that complementing it with more global topological requirements, we get more reliably computable access to ridge-like structures.

\medskip \noindent \textbf{Outline.}
Section~\ref{sec:2} formally defines the mid-sphere axis transform.
Section~\ref{sec:3} reviews the Morse theory of the squared distance function from a point.
Section~\ref{sec:4} introduces the notion of Faustian interchanges, which characterize the mid-sphere axis transform.
Section~\ref{sec:5} explains the pairing mechanism of persistent homology and its connection to Faustian interchanges.
Section~\ref{sec:6} exploits the vineyard algorithm to compute an approximation of the mid-sphere axis transform.
Section~\ref{sec:7} concludes the paper.

%%%%%%%%%%%%%%%%%%%%%%%%%%%%%%%%%%%%%%%%%%%%%%%%%%
%%%%%%%%%%%%%%%%%%%%%%%%%%%%%%%%%%%%%%%%%%%%%%%%%%
%% \newpage
\section{The Mid-sphere Axis Transform}
\label{sec:2}
%%%%%%%%%%%%%%%%%%%%%%%%%%%%%%%%%%%%%%%%%%%%%%%%%%
%%%%%%%%%%%%%%%%%%%%%%%%%%%%%%%%%%%%%%%%%%%%%%%%%%

As illustrated in Figure~\ref{fig:section}, a point may belong to the mid-sphere axis for more than one radius.
In other words, the projection of the mid-sphere axis transform in $\Rspace^4$ to the mid-sphere axis in $\Rspace^3$ is not necessarily injective, which is in contrast to the medial and the circum-sphere axes.

\smallskip
Let $M \subseteq \Rspace^3$ be a smoothly embedded connected closed surface, and write $S = S(x,r)$ for the sphere with center $x \in \Rspace^3$ and radius $r > 0$.
The set $\Rspace^3 \setminus M$ consists of two open components, a bounded one called the \emph{inside} of $M$, and an unbounded one called the \emph{outside} of $M$.
Similarly, we distinguish between the \emph{inside} of $S$, which is an open ball, and the \emph{outside} of $S$, which is the complement of a closed ball.
The sphere \emph{touches} the surface in a point $b \in M \cap S$ if the tangent planes of $M$ and $S$ at $b$ agree, and it touches $M$ \emph{generically} if $\sfrac{1}{r}$ is different from the two principal curvatures of $M$ at $b$.

\begin{figure}[hbt]
    \centering
    \vspace{0.0in}
    \resizebox{!}{2.2in}{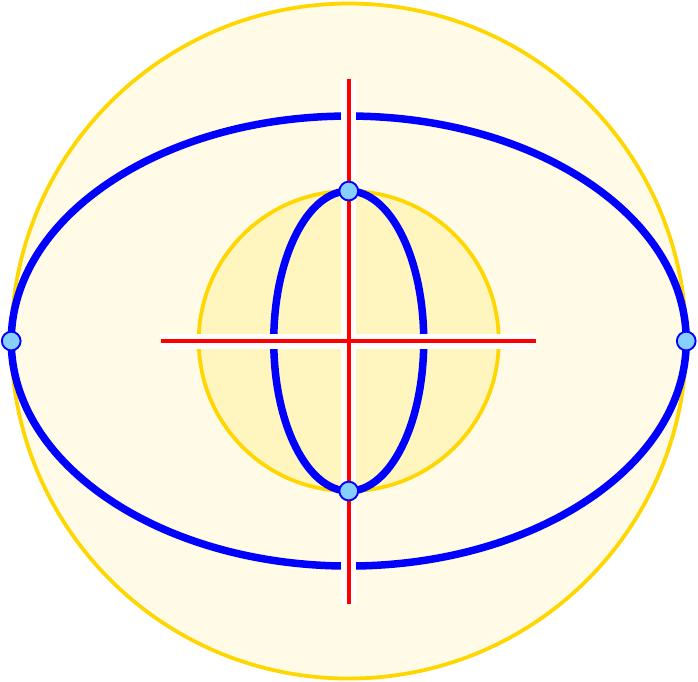}
    \vspace{-0.0in}
    \caption{\footnotesize Two locally cylindrical pieces of a \emph{blue} surface with elliptic cross-sections as shown.
    The smaller cylinder runs inside the wider cylinder, and both share an axis normal to the drawing plane.
    The two \emph{yellow} spheres touch the two cylinders in two points each, and because these spheres are concentric but have different radii, their common center belongs to two different \emph{red} sheets of the mid-sphere axis, which intersect in a line that runs along the common central axis of the two cylinders.}
    \label{fig:section}
\end{figure}

\begin{definition}
  \label{dfn:mid-sphere_axis_transform}
  Let $M \subseteq \Rspace^3$ be a smoothly embedded connected closed surface.
  A point-radius pair $(x,r)  \!\in \! \Rspace^3 \! \times \! \Rspace$ belongs to the \emph{mid-sphere axis transform} of $M$ if
  \begin{enumerate}[(i)]
    \item $S = S(x,r)$ touches $M$ generically in at least two points, $b \neq c$;
    \item there is a smooth closed curve, $\mu \subseteq M$, that passes through $b$ and $c$ and lies otherwise inside $S$;
    \item there is a smooth closed curve, $\sigma \subseteq S$, that passes through $b$ and $c$, and a filling surface, $\tau$, with boundary $\sigma$, that both lie otherwise inside $M$.
  \end{enumerate}
\end{definition}

\begin{figure}
    \centering
    \includegraphics[width=0.75\linewidth]{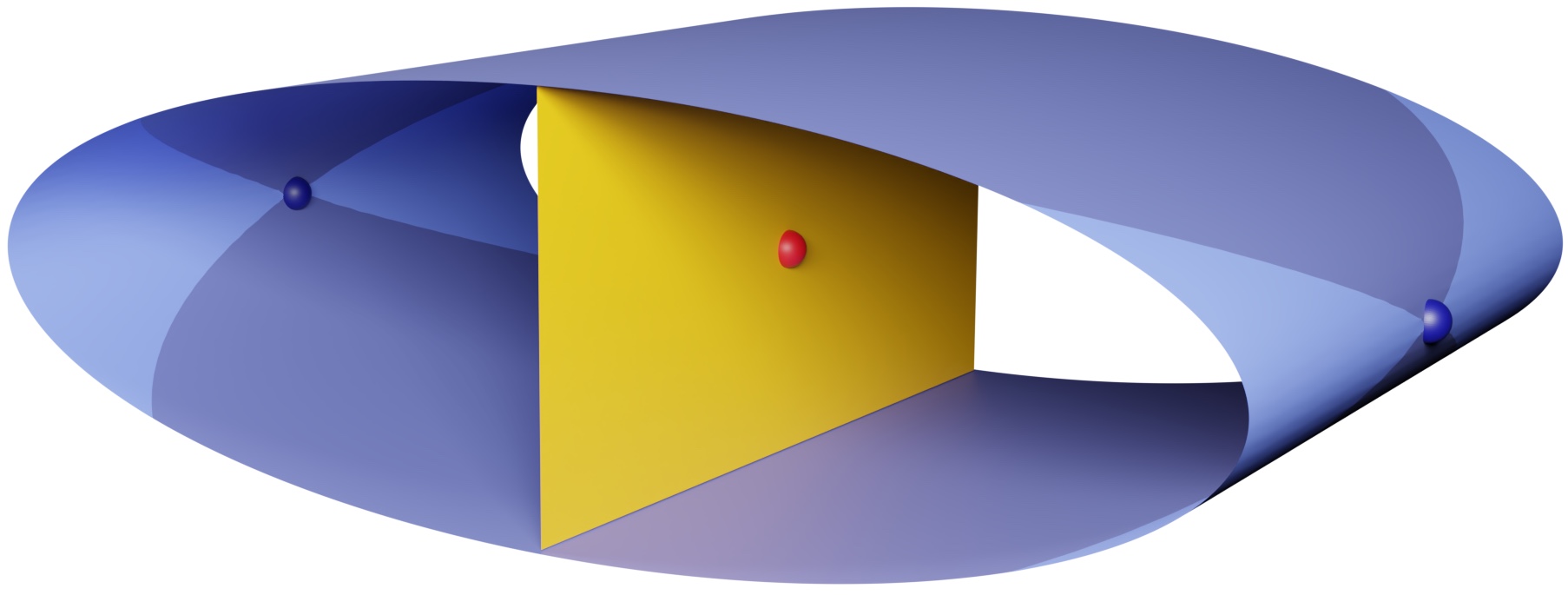}
    \caption{Another example of $S\cap M$, in which $M$ is a flattened cylinder and $S$ has its center marked in red on the yellow mid-sphere axis. 
    The points $b, c$ at which the sphere touches the surface are seen in dark blue.}
    \label{fig:def1}
\end{figure}
See Figure \ref{fig:EllipsoidA} for an intuitive example of the closed curves, $\mu$ and $\sigma$ and the filling surface, $\tau$.
In this example, $\tau$ is a disk, but more generally it may be an orientable $2$-manifold with boundary $\sigma$.
Its purpose is to separate the two open curves obtained by removing points $b$ and $c$ from $\mu$.
It is not necessarily simply-connected and thus may have handles, like the torus.
If $M$ is homeomorphic to a sphere, then Condition {\sf (iii)} can be weakened because the existence of $\tau$ follows from the existence of $\sigma$.
We write $\MSAxis{M}$ for the mid-sphere axis transform of $M$.
Generically, it is a $2$-dimensional surface in $\Rspace^4$.
The \emph{mid-sphere axis} is the projection to $\Rspace^3$ defined by forgetting the radii.
This projection is not necessarily an embedding, so the mid-sphere axis can have self-intersections.

\Skip{ 
  \smallskip \noindent \emph{Remark.}
  The topological requirements on the point-radius pairs of the mid-sphere axis transform can be strengthened by requiring that $b$ and $c$ be antipodal on $S$, or that $\sigma$ be a circle on $S$ and $\tau$ the disk whose boundary is this circle.
  We could also require that $\mu$ be part of the intersection of $M$ with a plane that passes through $b$ and $c$, and even that this plane be orthogonal to the disk $\tau$.
  We will not pursue any of these modifications to the definition in this paper, primarily because the geometric conditions are more difficult to recognize algorithmically.
} % END OF SKIP

%%%%%%%%%%%%%%%%%%%%%%%%%%%%%%%%%%%%%%%%%%%%%%%%%%
%%%%%%%%%%%%%%%%%%%%%%%%%%%%%%%%%%%%%%%%%%%%%%%%%%
%% \newpage
\section{Morse Theory of Squared Distance Function}
\label{sec:3}
%%%%%%%%%%%%%%%%%%%%%%%%%%%%%%%%%%%%%%%%%%%%%%%%%%
%%%%%%%%%%%%%%%%%%%%%%%%%%%%%%%%%%%%%%%%%%%%%%%%%%

The interaction of a surface, $M$, with spheres centered at $x \in \Rspace^3$ is conveniently described using the Euclidean distance of the points on $M$ from $x$.
This is a smooth function on $M$, unless $x \in M$, so it is preferable to use the square of this distance, which is always smooth.
We thus introduce $f_x \colon M \to \Rspace$ defined by $f_x (y) = \Edist{y}{x}^2$.
Writing $\nabla f_x (y)$ and $H_x (y)$ for the gradient and the Hessian of $f_x$ at $y \in M$, a point $y \in M$ is a \emph{critical point} of $f_x$ if $\nabla f_x (y) = 0$, and it is \emph{non-degenerate} if $H_x (y)$ is invertible.
The values of critical points are called \emph{critical values} of $f_x$, and all other values are \emph{non-critical values}.
Non-degenerate critical points are necessarily isolated, so if all critical points are non-degenerate, then there are only finitely many of them.
Indeed, $f_x$ is generically a \emph{Morse function} \cite{Mil63}, which is slightly stronger and requires that $f_x$ satisfies two conditions:
\begin{itemize}
  \item all critical points are non-degenerate;
  \item the values of the critical points are distinct.
\end{itemize}
Assuming $f_x \colon M \to \Rspace$ is Morse, there are only three types of critical points: \emph{minima}, \emph{saddles}, \emph{maxima}, which can be distinguished by the number of negative eigenvalues of the Hessian, which is $0$, $1$, $2$, in this sequence.
This number is traditionally referred to as the \emph{index} of the critical point; see Figure~\ref{fig:twan} where we see $2$ minima, $3$ saddles, and $1$ maximum of the displayed height function.

\begin{figure}[htb]
  \centering \vspace{0.1in}
  \resizebox{!}{1.45in}{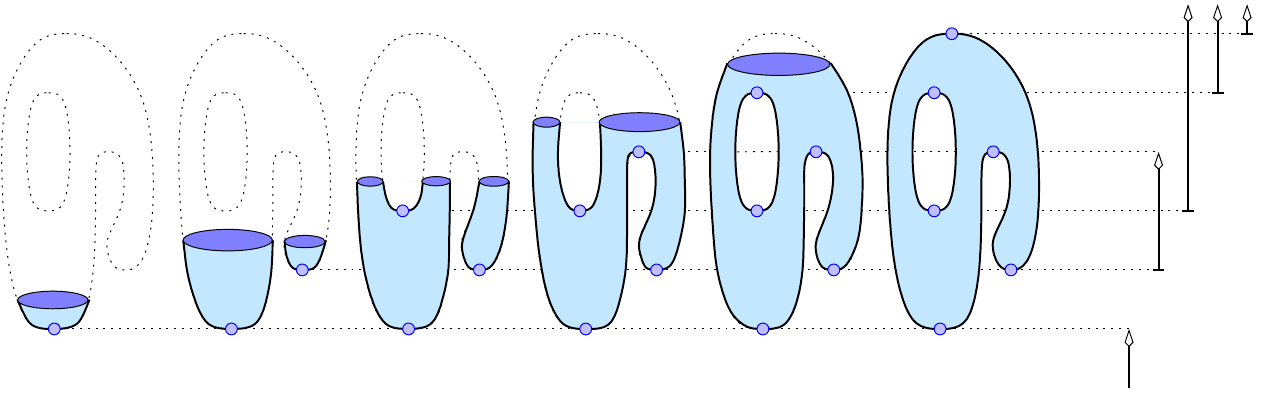}
  \caption{\footnotesize The six non-empty sublevel sets at non-critical values interleaved between the critical values of the height function on the torus-with-a-nose.
  The lowest and highest of the three saddles both give birth, while the middle saddle gives death.
  On the \emph{right}, we see the barcode that shows for which value a gap, loop, or closed surface exists in the sublevel set (see Section~\ref{sec:5} for definitions).}
  \label{fig:twan}
\end{figure}

\smallskip
Besides the three \emph{local} types, we also distinguish between two \emph{global} types of critical points.
For this, we write $M_r = f_x^{-1} [0,r^2]$ for the \emph{sublevel set} of $f_x$ at $r^2$, which contains all points whose Euclidean distance from $x$ is at most $r$.
A fundamental result in Morse theory \cite{Mil63} is that $M_r$ and $M_s$ have the same homotopy type if there is no critical values in the open interval $(r,s)$.
Furthermore, whenever $(r,s)$ contains a single critical value, then $M_r$ and $M_s$ have different homotopy types.
To quantify this difference, we use the ranks of the $p$-th reduced homology groups, which we denote $\Betti{p} (M_r)$.
Whenever we pass the value of a minimum, $\Betti{0}$ increases by $1$, unless this is the first minimum, in which case $\Betti{-1}$ decreases from $1$ to $0$.
Whenever we pass the value of a saddle, either two components merge, in which case $\Betti{0}$ decreases by $1$, or a component gains a loop, in which case $\Betti{1}$ increases by $1$.
Finally, whenever we pass the value of a maximum, $\Betti{1}$ decreases by $1$, unless this is the last maximum, in which case $\Betti{2}$ increases from $0$ to $1$.
This exhausts all cases and we note that each non-degenerate critical point changes only one rank by $1$.
We say the critical point gives \emph{birth} if the rank increases by $1$, and it gives \emph{death} if the rank decreases by $1$.

\smallskip
Both the local and the global types of the critical points of the squared distance function will be important to determine whether a point, $x$, with associated radius, $r$ belongs to the mid-sphere axis transform.

%%%%%%%%%%%%%%%%%%%%%%%%%%%%%%%%%%%%%%%%%%%%%%%%%%
%%%%%%%%%%%%%%%%%%%%%%%%%%%%%%%%%%%%%%%%%%%%%%%%%%
\section{Faustian Interchanges}
\label{sec:4}
%%%%%%%%%%%%%%%%%%%%%%%%%%%%%%%%%%%%%%%%%%%%%%%%%%
%%%%%%%%%%%%%%%%%%%%%%%%%%%%%%%%%%%%%%%%%%%%%%%%%%

Consider now a continuous $1$-parameter family (a \emph{curve}) of smooth functions that starts with a Morse function on $M$ and ends with another Morse function on $M$.
Unless these two Morse functions are very similar, the family necessarily passes through functions that are not Morse.
By the \emph{essential property} of Cerf theory \cite{Cer68}, it is possible to choose the family such that it passes through only a finite number of functions that are not Morse, and each such function has only one violation of the two conditions that characterize a smooth function as Morse:
\begin{enumerate}[I.]
  \item two critical points that share the same value;
  \item one degenerate critical point with simple local neighborhood.
\end{enumerate}
A violation II is the transient configuration that appears when two critical points collide and in the process annihilate each other.
We refer to it as a \emph{cancellation} of the two critical points.
Its inverse is an \emph{anti-cancellation}, which creates two critical points where there were none.
We are however more interested in a violation I, which is the transient configuration that appears when two critical points swap the order of their values.
We refer to it as an \emph{interchange}.

\smallskip
An interchange of two minima does not affect their global type, unless they are the first two minima (the ones with smallest value).
Then the first minimum becomes second and changes from giving death to giving birth, and the other way around for the second minimum.
This happens where the curve of squared distance functions crosses the medial axis, because the point at the crossing has two closest points on $M$ or, equivalently, the squared distance function from the crossing has two minima with smallest value.
\begin{definition}
  \label{dfn:Faustian_interchange}
  A \emph{Faustian interchange} is an interchange of two critical points during which they swap their global types, one from giving death to giving birth and other from giving birth to giving death.
\end{definition}
Critical points that swap their global types when they interchange necessarily have the same index; see Section~\ref{sec:5}.
As described above, the Faustian interchanges of two minima correspond to points of the medial axis, and by a similar argument, the Faustian interchanges of two maxima correspond to points of the circum-sphere axis.
Most important for this paper are the Faustian interchanges of two saddles, which correspond to points of the mid-sphere axis.
\begin{theorem}
  \label{thm:equivalence}
  A point $x \in \Rspace^3$ with radius $r > 0$ belongs to the mid-sphere axis transform of $F$ iff $f_x$ has a Faustian interchange of two saddles at $r^2$.
\end{theorem}
\begin{proof}
  ``$\implies$'':
  Assuming $(x,r) \in \MSAxis{M}$, let $b, c \in M$ be the points at which $S = S(x,r)$ touches $M$, and let $\mu, \sigma$ be the two closed curves and $\tau$ the surface that fills $\sigma$ whose existence is guaranteed by Conditions~(ii) and (iii) of Definition~\ref{dfn:mid-sphere_axis_transform}.
  Let $X$ be the connected component of $M_r$ that contains $b$ and $c$.
  Since $\mu$ exists, $X \setminus \{b\}$ is connected, and so is $X \setminus \{c\}$.
    %% ABOVE SENTENCE REPLACES:
    %% The connectivity of $X \setminus \{b\}$ and $X \setminus \{c\}$ is clear from the existence of $\mu$.
  It is also clear that $b$ and $c$ cannot be minima or maxima, so they are saddles.
  If $X \setminus \{b,c\}$ is connected, then there is a curve that starts at one branch of $\mu \setminus \{b,c\}$ and ends at the other whose points lie on or inside $S$.
  But this curve contradicts the existence of the curve $\sigma$ and the surface $\tau$ that fills $\sigma$.
  Hence, if we add $b$ to $X \setminus \{b,c\}$ before $c$, this first gives death to a component and then birth to a loop, and if we add $c$ to $X \setminus \{b,c\}$ before $b$, we observe the same pattern.
  Hence, $f_x$ has a Faustian interchange of $b$ and $c$ at $r^2$.

  \smallskip
  ``$\impliedby$'':
  Let $b$ and $c$ be the two saddles of the Faustian interchange of $f_x$ at $r^2$.
  Since $b$ and $c$ swap their global types, they must belong to the same component of the corresponding sublevel set, $X \subseteq M_r$, this component contains a loop that passes through $b$ and $c$, and $X \setminus \{b,c\}$ is not connected.
  This loop is the closed curve $\mu$ required by Condition~(ii), and that $X \setminus \{b,c\}$ is not connected implies the existence of the curve $\sigma \subseteq S$ and the surface $\tau$ that fills $\sigma$ inside $M$, as required by Condition~(iii).
\end{proof}

%%%%%%%%%%%%%%%%%%%%%%%%%%%%%%%%%%%%%%%%%%%%%%%%%%
%%%%%%%%%%%%%%%%%%%%%%%%%%%%%%%%%%%%%%%%%%%%%%%%%%
\section{Pairing of Critical Points}
\label{sec:5}
%%%%%%%%%%%%%%%%%%%%%%%%%%%%%%%%%%%%%%%%%%%%%%%%%%
%%%%%%%%%%%%%%%%%%%%%%%%%%%%%%%%%%%%%%%%%%%%%%%%%%

The idea of persistent homology is that the critical points of a function can be paired to delimit the ranges in which the corresponding features appear in the sublevel sets of the function; see e.g.\ \cite{EdHa10}.
This additional structure is useful in two ways.
First, it provides a means to assess how robust a Faustian interchange is, and we will introduce several ways to quantify this notion.
Second, it leads to an algorithm to construct the mid-sphere axis transform.

\smallskip
To explain the pairing, we recall that each critical point gives either birth or death, namely to a homology class of the sublevel set at the corresponding critical value.
Since we use reduced homology, the empty sublevel set has a $(-1)$-dimensional class, which dies when the first minimum is encountered.
Later minima give birth to $0$-dimensional classes, which die at the hand of the saddles that merge them with other $0$-dimensional classes born before them.
This is the \emph{elder rule}, which prescribes that the older class survives by absorbing the younger class.
Alternatively, a saddle may give birth to a $1$-dimensional class, which later dies at the hand of a maximum.
Indeed, every maximum gives death to such a class, except the last maximum, which gives birth to a $2$-dimensional class, namely that of the entire surface.
Besides the last maximum, there are two saddles per genus of the closed surface that remain unpaired in this process, and they represent the homology of $M$.

\begin{figure}[hbt]
  \centering \vspace{0.1in}
  \resizebox{!}{1.2in}{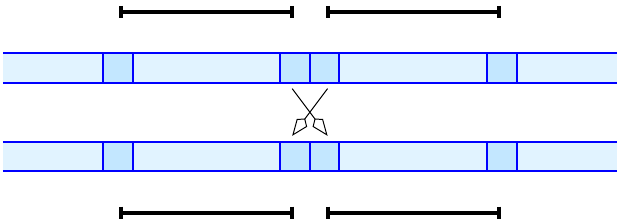}
  \caption{\footnotesize The ordered sequences of critical points before and after a Faustian interchange of two saddles.
  Note that the two saddles swap their partners under the pairing.}
  \label{fig:Faustian}
\end{figure}

\smallskip
To visualize this structure, we list the critical points from left to right in the order of their values and mark a pair $(a,b)$ by the interval that begins at $a$ and ends at $b$, writing $a = \Birth{b}$ and $b = \Death{a}$.
The index of $b$ is necessarily one higher than that of $a$, so there are only three kinds of (finite) intervals for a function on a $2$-dimensional surface: those with endpoints of index $-1,0$, of index $0,1$, and of index $1,2$.
Letting $b, c$ be two saddles of a Faustian interchange, we therefore have an interval from $a = \Birth{b}$ to $b = \Death{a}$ and another interval from $c = \Birth{d}$ to $d = \Death{c}$, as illustrated in Figure~\ref{fig:Faustian}.
Since $b$ and $c$ swap their global types, they also swap their partners in the pairing; that is: $a = \Birth{c}, c = \Death{a}$ and $b = \Birth{d}, d = \Death{b}$ after the interchange.
Note however that interchanging two saddles does not necessarily make a Faustian interchange.
Assuming a Faustian interchange as illustrated in Figure~\ref{fig:Faustian}, we consider the two intervals and call
\begin{align}
  \Premed{x}{b}{c}  &= f_x (b) - f_x (a) = r - \Edist{x}{a} , \\
  \Postmed{x}{b}{c} &= f_x (d) - f_x (c) = \Edist{x}{d} - r 
\end{align}
the \emph{pre-meditation} and \emph{post-meditation} of the Faustian interchange, in which $r = \Edist{x}{b} = \Edist{x}{c}$, $a$ is the first point of the second component of $X \setminus \{b,c\}$ encountered by the growing sphere centered at $x$, and $d$ is the last point of the first component of $M \setminus \mu$ encountered by that sphere.
We will return to these notions when we discuss the pruning on the mid-sphere axis in the next section.

%%%%%%%%%%%%%%%%%%%%%%%%%%%%%%%%%%%%%%%%%%%%%%%%%%
%%%%%%%%%%%%%%%%%%%%%%%%%%%%%%%%%%%%%%%%%%%%%%%%%%
%% \newpage
\section{Computation and Examples}
\label{sec:6}
%%%%%%%%%%%%%%%%%%%%%%%%%%%%%%%%%%%%%%%%%%%%%%%%%%
%%%%%%%%%%%%%%%%%%%%%%%%%%%%%%%%%%%%%%%%%%%%%%%%%%

This section uses the characterization of the mid-sphere axis in terms of Faustian interchanges (Theorem~\ref{thm:equivalence}) to describe an algorithm that constructs a staircase approximation of this axis.
To simplify the discussion, we ignore the radii, which may be computed without much extra effort.
The algorithm is discrete and extends to the construction of the medial and the circum-sphere axes.

\medskip \noindent \textbf{The algorithm.}
We assume the input surface is given as a triangulation; that is: a $2$-dimensional simplicial complex whose vertices are locations in $\Rspace^3$ such that each edge belongs to exactly two triangles and each vertex belongs to an alternating cyclic sequence of edges and triangles.
Given a point $x \in \Rspace^3$, we assign $f_x (u) = \Edist{u}{x}^2$ to each vertex, $u$, and the maximum value of the two or three incident vertices to each edge and triangle.
Sorting the vertices, edges, and triangles by value and breaking ties by dimension, we get the filter for which we compute the pairing as defined by persistent homology; see e.g.\ \cite[Chapter VII]{EdHa10}.
In this setting, the vertices, edges, and triangles play the roles of the minima, saddles, and maxima of the squared distance, respectively, and the transposition of two edges---one paired with a vertex and the other with a triangle---is a Faustian interchange if they replace each other in their respective pairs and thus alter their global types.

\smallskip
However, only a measure-zero set of points $x$ in $\Rspace^3$ induce squared distance functions that catch a Faustian interchange in the act.
We therefore use the vineyard algorithm introduced in \cite{CEM06} to probe space with short segments of filters to find points of the mid-sphere axis.
Given the endpoints $x', x'' \in \Rspace^3$ of such a segment, this algorithm maintains the filter and the pairing of simplices while continuously moving $x$ from $x'$ to $x''$.
The elementary operations in this sweep are transpositions of contiguous simplices in the filter, which are scheduled according to the linear interpolation between their initial and final values.
This amounts to running bubble sort on the filter at $x'$ using the filter at $x''$ as the target ordering.
Whenever a transposition swaps two edges, we check whether one gives birth, the other gives death, and the swap changes their global types.
If yes, then we have found a point whose squared distance function witnesses a Faustian interchange.

\smallskip
We use segments connecting neighboring points of the integer lattice, $\Zspace^3 \subseteq \Rspace^3$.
For example, we run the vineyard algorithm for the segment from $x' = (i,j,k)$ to $x'' = (i+1,j,k)$, and add the square $(i+\frac{1}{2}) \times [j \pm \frac{1}{2}] \times [k \pm \frac{1}{2}]$ to the staircase approximation of the mid-sphere axis iff there is at least one Faustian interchange along this segment.
We begin the hunt for Faustian interchanges at an arbitrary point of $\Zspace^3$ and reduce the boundary matrix of the corresponding filter from scratch.
Thereafter, we run the vineyard algorithm along the segments to the six neighboring integer points.
When we arrive at the end of one of these segments, we repeat while reusing the final reduced matrix instead of computing it from scratch and making sure that no segment is swept twice.
Sample results of this algorithm are shown in Figure~\ref{fig:axes}. 
\begin{figure}[h!]
  \vspace{0.0in}
  \centering
  \includegraphics[width = 0.37\textwidth]{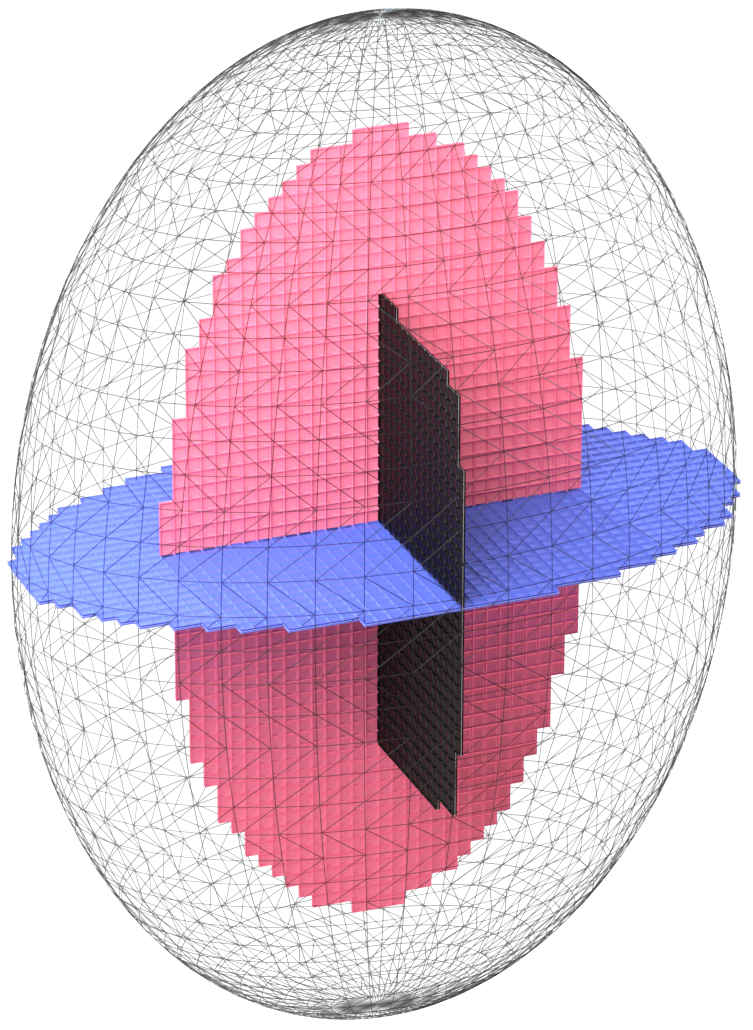}
  \includegraphics[width = 0.6\textwidth]{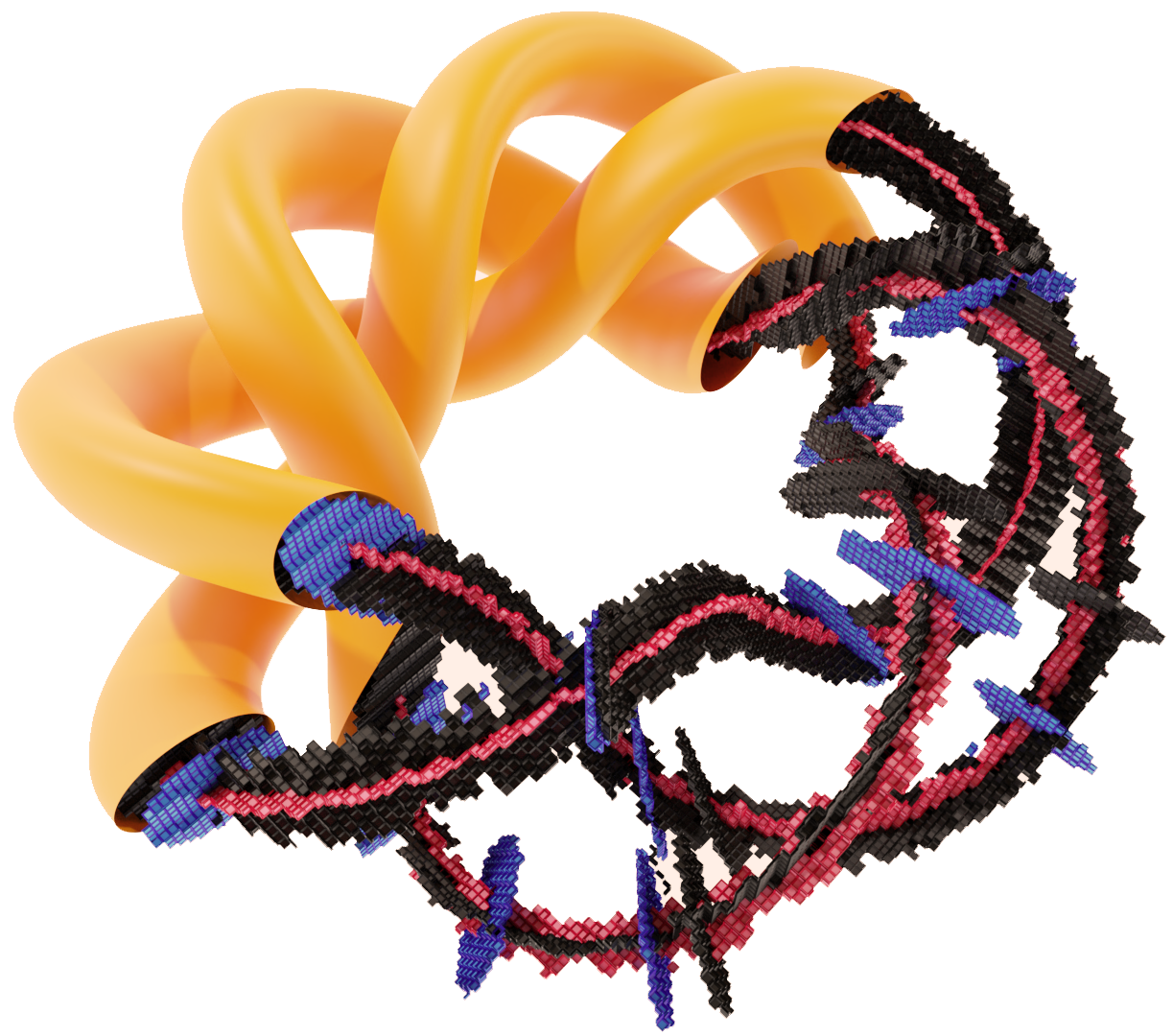}
    \vspace{-0.0in}
    \caption{\footnotesize The \emph{pink} medial, \emph{black} mid-sphere, and \emph{blue} circum-sphere axes of an ellipsoid on the \emph{left} and an elliptically thickened $(3,7)$-torus knot on the \emph{right}.
    To improve the visualization, we restrict ourselves to points inside the surfaces, which explains why the mid- and circum-sphere axes do not extend outside the two surfaces.
    }
    \label{fig:axes}
\end{figure}

\smallskip
Two or more segments can be swept in parallel, so we get some mileage out of partitioning the segments connecting neighboring integer points into a few sets, which are then processed in parallel.
\Skip{
  However, no matter how we schedule the segments, the number of swaps remains large.
  To see this note that two vertices swap along a segment iff it intersects the perpendicular bisector of the pair.
  Assuming we explore all segments connecting neighboring integer points in $[0,g]^3$, the number of such segments is $O(g^2)$.
  Letting $n$ be the number of vertices in the triangulation of the surface, this amounts to a total of $O (g^2 n^2)$ swaps.}
However, no matter how we schedule the segments, the number of interchanges is large.
Focusing on the vertices, and letting $n$ be their number and $[0,g]^3$ the domain,
we can na\"ively bound the number of interchanges by $O(g^3n^2)$, 
where every pair of vertices is interchanged along every segment.
To improve this bound, we count the number of segments along which a fixed pair is interchanged. 
This happens when their positions trade place in the ordering,
so any segment that crosses the perpendicular bisector of the pair
contributes to exactly one such interchange.
There are $O(g^2)$ such segments, so summing over all pairs of vertices gives at most $O(g^2n^2)$ interchanges.

\medskip \noindent \textbf{The matrices.}
The efficient implementation of the algorithm is challenging, in particular the computation of the reduced boundary matrices.
Assuming familiarity with the vineyard algorithm described in \cite{CEM06}, we sketch some of the more important design decisions that help in this respect.
Following \cite{CEM06}, we use a sparse matrix representation tailored to efficiently support the maintenance operations in the vineyard algorithm.
For convenience, we use $\Zspace / 2 \Zspace$ coefficients so homology can be computed using boolean matrices.
To reduce the boundary matrix, the algorithm swaps rows and columns, and it adds columns to each other.
We thus store the matrix as a sequence of sparse columns, and we maintain permutations $\pi_R$ and $\pi_C$ for the rows and columns, respectively.
With these permutations, it is easy to swap two rows or two columns, but they add an indirection when we access entries: the entry in row $i$ and column $j$ is physically stored in row $\pi_R (i)$ and column $\pi_C (j)$.

\smallskip
Each column is an ordered list of non-zero entries.
Adding two columns means computing the \texttt{xor} function, and since the columns are sorted, this can be done by walking along the two lists in parallel and ignoring rows that appear in both.
To decide which columns to add, the algorithm needs access to the maximum row index of the non-zero entries in a column.
Because of the permutation $\pi_R$, this row is not necessarily the last in the list.
Profiling an earlier version of our software revealed that finding this row index in a linear scan was the main bottleneck, so we decided to explicitly record this row index for each column and to update the record whenever it changes.
We mention that the vineyard algorithm also adds rows to each other in a second matrix, which we support by storing the transpose of that matrix in the aforementioned format.

\smallskip
In spite of using sparse versions, the amount of memory needed to remember reduced matrices is significant, so we dispose of them as soon as we can while still avoiding the reduction of the boundary matrix from scratch.
This is where breadth-first search in scheduling the segments shines, as its front of unfinished integer points is much smaller than for example for depth-first search.

\medskip \noindent \textbf{Pruning strategies.}
Similar to the medial axis, the mid-sphere axis is sensitive to small fluctuations in the data, and since our algorithm is approximate, this tends to give rise to artifacts at small scales.
We therefore measure scale and prune artifacts if their scale is found to be below a fixed threshold.
We distinguish between two strategies: the first uses a notion of \emph{distance} between the two interchanging simplices, and the second assesses their \emph{persistence}.

\smallskip
In \emph{Euclidean pruning}, we check the Euclidean distance between the two interchanging simplices.
For edges and triangles, we use the Euclidean distance between their midpoints and barycenters, respectively.
See Figure~\ref{fig:hexagons} for an illustration of the effect for different thresholds.
A likely improvement would be using the geodesic distance along the surface, but we get satisfactory results with the Euclidean distance, which we use for all three types of axes.
\begin{figure}[h!]
  \centering
  \vspace{-0.1in}
  \includegraphics[width = \textwidth]{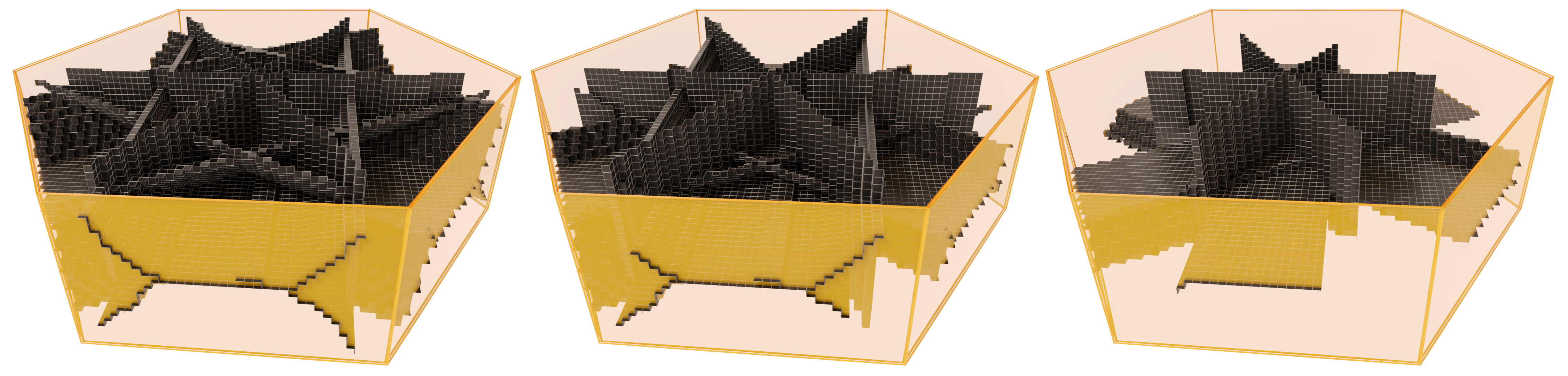}
    \vspace{-0.2in}
    \caption{\footnotesize A hexagonal solid with its mid-sphere axis progressively modified by Euclidean pruning.
    The top and bottom faces of the solid are removed for improved visibility.
    }
    \label{fig:hexagons}
\end{figure}
Alternatively, we may measure the distance combinatorially, by counting the simplices needed to connect the interchanging simplices.
We use only the simplest incarnation of this idea.
In \emph{face pruning}, we make sure that the interchanging simplices do not have a common face, and we apply this constraint to the construction of the mid-sphere and the circum-sphere axes.
Likewise, in \emph{coface pruning}, we make sure that they do not have a common coface, and we apply this to the construction of the medial and the mid-sphere axes.

\smallskip
In \emph{persistence pruning}, we consider the persistence of the two interchanging simplices; that is: their pre- and post-meditations, as defined in Section~\ref{sec:5}.
We note that they have different geometric meaning: for the mid-sphere axis, the pre-meditation measures the two components inside the mid-sphere before they merge due to adding the first of the two edges, while the post-meditation measures the loop formed by adding both edges simultaneously.
In our current implementation, we require that both meditations exceed a fixed threshold.
We find persistence pruning indispensable for computing the mid-sphere axis, as it detects many spurious loops that form the boundaries of triangles and are filled right after being formed.
More generally, we use it for all three types of axes.

\begin{figure}[h!]
  \centering
  \vspace{-0.1in}
  \includegraphics[width = .95\textwidth]{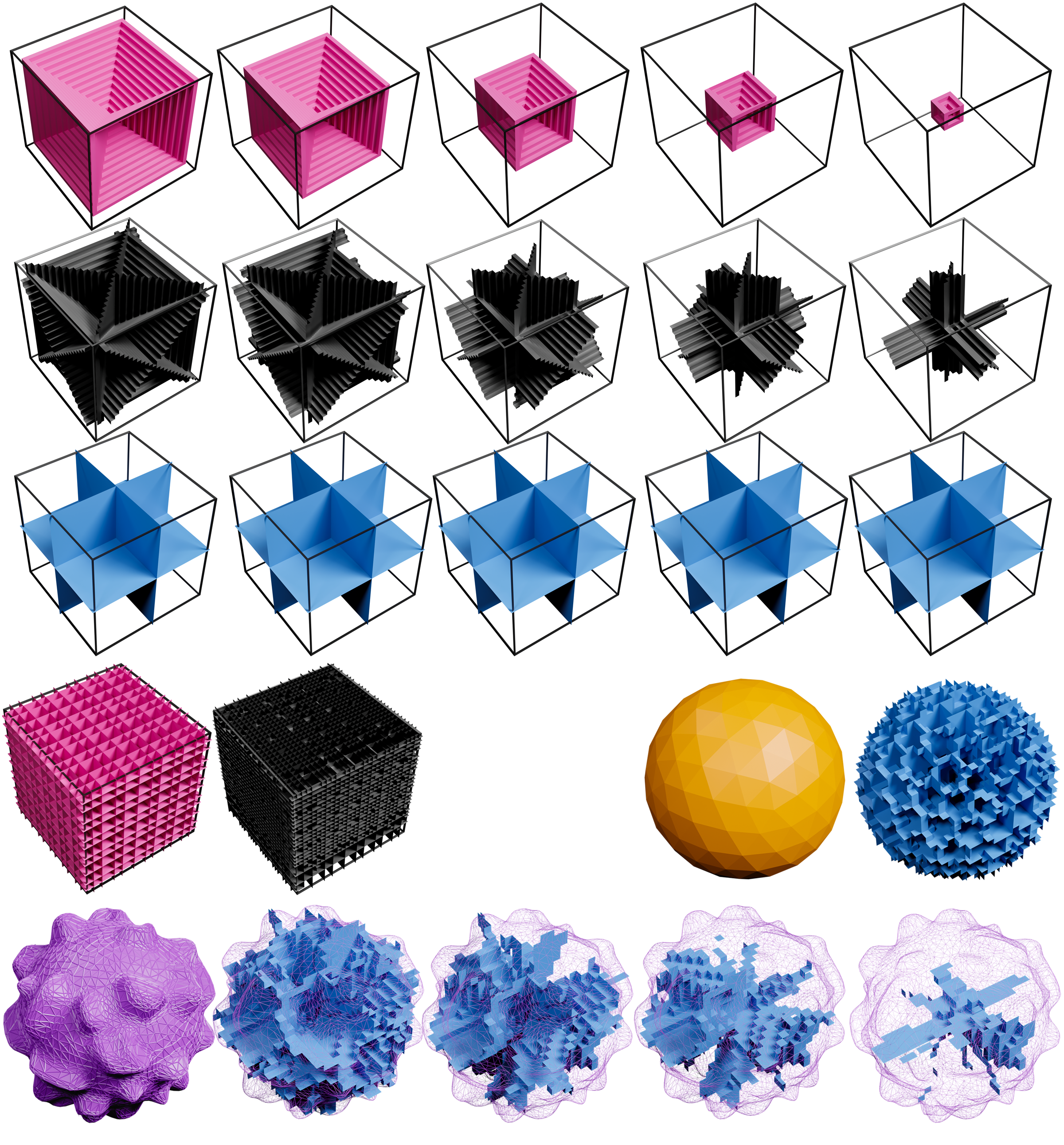}
    \vspace{-0.2in}
    \caption{\footnotesize Here we demonstrate some effects of pruning on the medial (pink), midsphere (black), and circumsphere (blue) axes. The top three rows demonstrate the effect of Euclidean pruning on the axes resulting from a cube. We also see that pruning is sometimes not needed for the circumsphere axis.
    In the fourth row, we demonstrate the necessity of the other types of pruning. Without coface pruning, the medial axis of the cube blows up, as does the midsphere axis without persistence pruning. To show the circumsphere axis behaving similarly, we needed a different input than the cube, and so used the sphere (in yellow). Without face pruning, it too becomes unwieldy. 
    In the bottom row, we demonstrate an input (the purple virus-like shape) for which Euclidean pruning has an impact on the circumsphere axis. 
    }
    \label{fig:minisquares}
\end{figure}

\medskip \noindent \textbf{Code availability.}
We have published the code used to compute these examples on GitHub \footnote{\underline{\href{https://github.com/medial-ax/medial-ax}{https://github.com/medial-ax/medial-ax}}}. It can be compiled and run locally, or sampled via a live demo we have running via GitHub Pages \footnote{\underline{\href{https://medial-ax.github.io/medial-ax/}{https://medial-ax.github.io/medial-ax/}}}.

% END OF SKIP
%%%%%%%%%%%%%%%%%%%%%%%%%%%%%%%
%%%%%%%%%%%%%%%%%%%%%%%%%%%%%%%

%%%%%%%%%%%%%%%%%%%%%%%%%%%%%%%%%%%%%%%%%%%%%%%%%%
%%%%%%%%%%%%%%%%%%%%%%%%%%%%%%%%%%%%%%%%%%%%%%%%%%
%% \newpage
\section{Discussion}
\label{sec:7}
%%%%%%%%%%%%%%%%%%%%%%%%%%%%%%%%%%%%%%%%%%%%%%%%%%
%%%%%%%%%%%%%%%%%%%%%%%%%%%%%%%%%%%%%%%%%%%%%%%%%%

The main contribution of this paper is the introduction of the mid-sphere axis transform of a smoothly embedded connected closed surface in $\Rspace^3$, which is sensitive to ridge-like structures on the surface.
The definition is in terms of local and global topological aspects of the family of squared distance functions on the surface from a point.
We also introduce the notion of a Faustian interchange, which offers a unified framework in which the medial and mid-sphere axis transforms are particular cases.
This framework extends beyond three dimensions and leads to $d$ axis transforms for any smoothly embedded hypersurface in $\Rspace^d$.

\smallskip
The presented algorithm serves as a proof of concept, and improved algorithms are desirable.
One promising idea is the extension of the Voronoi heuristic to the mid-sphere axis. 
Note that the Voronoi tessellation is already used to compute the standard medial axis \cite{ABE09}, and, by analogy, the furthest-point Voronoi tessellation yields the circum-sphere axis.
To extend this heuristic to the mid-sphere axis, we will have to collect and combine substructures from various higher order Voronoi tessellations of the same set of points.

\smallskip
Another interesting direction is the analysis of effective pruning strategies that lower the sensitivity to noise in the data by defining mid-sphere axes on different scale levels.
For example, can the $\lambda$-medial axis introduced by Chazal and Lieutier~\cite{ChLi05} be extended to the mid-sphere axis?

%%%%%%%%%%%%%%%%%%%%%%%%%%%
% Temp. place the
% \clearpage
\setcounter{section}{0}
\renewcommand{\thesection}{\appendixname~\Alph{section}}
\section{Pseudocode}
% \ES{todo: add in persistence pruning more explicitly? maybe also other pruning?}
Here we present pseudocode of our algorithm. This is written per homology dimension and per component of the grid. In our implementation, we run different grid components in parallel. 
\begin{algorithm}
\caption{Computing the mid-sphere axis}
\begin{algorithmic}
\Function{Initialize}{$v$}
    \State $\pi_v \gets$ \Call{Ordering}{$v$}
    \Comment{Order simplices by distance function from grid point $v$}
    \State $D_v \gets$ boundary matrix ordered by $\pi_v$
    \State $R_v, U_v \gets$ \Call{Reduce}{$D_v$}\Comment{Following \cite{CEM06} we have $D_v = R_vU_v$}
\EndFunction

\Function{ProcessSegment}{$u, v$}
    \State{$F_{u,v} \gets \{\}$}
    \State{
        $D_v\gets D_u$,
        $R_v\gets R_u$,
        $U_v\gets U_u$
    }
    \Comment{Copy matrices. These are about to be updated}
    \State{$\pi_u \gets$ \Call{Ordering}{$u$}, $\pi_v \gets$ \Call{Ordering}{$v$}}
    \Comment{Simplex ordering wrt. $u$ and $v$}
    \State{$I\gets $ \Call{BubbleInterchanges}{$\pi_u$, $\pi_v$}}
    \Comment{Interchanges required to go from $\pi_u$ to $\pi_v$}
    \For{$(i, j)$ in $I$}
        \Comment{Handle interchange of simplices $\sigma_i$ and $\sigma_j$}
        \State{\Call{Vineyards}{$i$, $j$, $D_v$, $R_v$, $U_v$}}
        \Comment{Adjust matrices for $v$ following \cite{CEM06}}
        \If{was Faustian}
            \Comment{ See \autoref{dfn:Faustian_interchange} }
            \State{$F_{u,v} \gets F_{u,v} \cup \{(i, j)\}$}
            \Comment{Register $(i,j)$ as a Faustian interchange}
        \EndIf
    \EndFor
\EndFunction

\Function{Main}{$v_0$}
    \Comment{We start reducing from any vertex $v_0$}
    \State \Call{Initialize}{$v_0$}
    \For{$(u, w) \in$ \Call{BFS}{$v_0$}}\Comment{Breadth-first traversal}
        \State \Call{ProcessSegment}{$u$, $w$}
        \If{$\{u,v\}$ is visited for the last time}\Comment{Eager cleanup to save memory}
            \State Delete $D_w$, $R_w$, $U_w$ where $w=u$ or $w=v$ (or both)
        \EndIf
    \EndFor
    \For{$(u,v)$ with entries in $F_{u,v}$}
        \State{\Output{\Call{Prune}{$F_{u,v}$}}}
    \EndFor
\EndFunction

\end{algorithmic}
\end{algorithm}
\clearpage

%%%%%%%%%%%%%%%%%%%%%%%%%%%
%% \newpage

\end{document}